\documentclass[a4paper]{article}
\usepackage{amsthm,amssymb,amsmath,enumerate}
\usepackage{algorithm}

\usepackage{tikz}
\usetikzlibrary{backgrounds}

\tikzstyle{hvertex}=[thick,circle,inner sep=0.cm, minimum size=2.3mm, fill=white, draw=black]
\tikzstyle{hedge}=[very thick]
\tikzstyle{rededge}=[very thick,red]
\tikzstyle{point}=[draw,circle,inner sep=0.cm, minimum size=1mm, fill=black]
\tikzstyle{pointer}=[thick,->,shorten >=2pt,color=hellgrau]
\tikzstyle{facebdry}=[color=auchblau, very thick] % face boundary
\tikzstyle{face}=[facebdry,fill=hellblau]
\tikzstyle{nface}=[color=hellblau,fill=hellblau,thick] % naked face, without boundary

\usepackage{pgfplots}
\pgfplotsset{compat=1.13}

\usepackage{pgfplotstable}
\usepackage{booktabs}
\usepackage{array}
\usepackage{colortbl}

\pgfplotstableset{% global config, for example in the preamble
every head row/.style={before row=\toprule,after row=\midrule},
every last row/.style={after row=\bottomrule}
}

\colorlet{auchblau}{blue!60!white}
\colorlet{hellblau}{blue!20!white}
\colorlet{hellrot}{red!40!white}
\colorlet{hellgrau}{black!30!white}

\newtheorem{definition}{Definition}
\newtheorem{proposition}[definition]{Proposition}
\newtheorem{theorem}[definition]{Theorem}
\newtheorem{corollary}[definition]{Corollary}
\newtheorem{lemma}[definition]{Lemma}

\newcommand{\1}{{\rm 1\hspace*{-0.4ex}%
\rule{0.1ex}{1.52ex}\hspace*{0.2ex}}}

\newcommand{\comment}[1]{}

\newcommand{\emtext}[1]{\text{\em #1}}

\newcommand{\mymargin}[1]{% 
  \marginpar{%
    \begin{minipage}{\marginparwidth}\small%
      \begin{flushleft}%
        #1%
      \end{flushleft}%
   \end{minipage}%
  }%
}%

\newcommand{\sm}{\setminus}

\newcommand{\prob}[1]{\textsf{#1}}
\newcommand{\alg}[1]{\textsf{#1}}

\usepackage{algpseudocode}
\newcommand{\inp}{\smallskip\makebox[1.2cm][l]{\em Input}}
\newcommand{\out}{\makebox[1.2cm][l]{\em Output}}

\usepackage{framed}
\usepackage{boxedminipage}

\newlength{\innerboxwidth}
\newenvironment{XXhproblem}[1]{
\smallskip\noindent
\begin{boxedminipage}{\textwidth}\vspace*{3pt}
\prob{#1}\\[6pt]
\hspace*{1.3cm}\begin{minipage}{\innerboxwidth}
\begin{itemize}\labelwidth3cm
}{
\end{itemize}\end{minipage}\vspace*{3pt}
\end{boxedminipage}
}

\newcommand{\instance}{
\item[{\makebox[1.4cm][l]{\em Instance}}]
}
\newcommand{\task}{
\item[{\makebox[1.4cm][l]{\em Task}}]
}

\newenvironment{hproblem}[1]{
\begin{XXhproblem}{#1}
}{
\end{XXhproblem}
}

\newcommand{\opt}{\text{OPT}}
\newcommand{\spa}{\text{SPRS}}
\newcommand{\Tspa}{\text{SPRS}^2}
\newcommand{\bigO}{O}

\title{Fast Algorithms for Delta-Separated Sparsity Projection}
\author{Henning Bruhn \and Oliver Schaudt}
\date{}

\begin{document}
\maketitle

\begin{abstract}
\noindent
We describe a fast approximation algorithm for the $\Delta$-separated sparsity projection problem.
The $\Delta$-separated sparsity model was introduced by Hegde, Duarte and Cevher (2009) to capture
 the firing process of a single Poisson neuron with absolute refractoriness.
The running time of our projection algorithm is linear for an arbitrary (but fixed) precision and it is both a head and a tail approximation.
This solves a problem of Hegde, Indyk and Schmidt (2015).

We also describe how our algorithm fits into the
approximate model iterative hard tresholding framework of Hegde, Indyk and Schmidt (2014)
that allows to recover $\Delta$-separated sparse signals from noisy random linear measurements.
The resulting recovery algorithm is substantially faster than the existing one, at least for large data sets.
\end{abstract}

\section{Introduction}

Compressed sensing is based on the insight that real-life signals are often sparse.
In character recognition, for example, most pixels are white or nearly so.
That premise allows to obtain a signal with far fewer measurements than classic lower bounds suggest.
It necessitates, however, a computational step that recovers the signal from those measurements. 

Fast iterative recovery algorithms, such as \emph{iterative hard tresholding} (IHT)~\cite{BLUMENSATH2009265} and \emph{compressed sampling matching pursuit} (CoSaMP)~\cite{NEEDELL2009301}, have been proposed. A common step in these 
algorithms is the so-called \emph{sparsity projection} or \emph{model approximation} step: to project a signal onto a sparse signal.

More concretely, consider the most basic notion of sparsity: 
a vector is \emph{$k$-sparse} if at most $k$ of its entries are non-zero.
Then the projection problem consists in finding, for given $x\in\mathbb R^n$,
 a $k$-sparse $y\in\mathbb R^n$ such that $||x-y||_2^2$ is minimized.
This problem is almost trivial---keeping the $k$ largest entries of $x$ and overwriting all other entries with $0$ will do.

We study a slightly more elaborate sparsity model: the  \emph{$\Delta$-separated sparsity model}.
A vector $y$ is \emph{$\Delta$-separated} if 
any two non-zero entries $y_i,y_j$ satisfy $|i-j|\ge \Delta$ or $i=j$.
We treat the corresponding  projection problem:
given $x \in \mathbb R^n$, $k$ and $\Delta$, 
find the $k$-sparse $\Delta$-separated vector $y$ minimizing $||x-y||_2^2$.

The $\Delta$-separated sparsity model was introduced by Hegde, Duarte and Cevher~\cite{HDC09} 
to capture neuronal spike trains in the framework of compressed sensing.
They observed that the projection on the $\Delta$-separated sparsity model can be formulated as an integer linear program over a totally unimodular constraint matrix.
Interior point methods, then, allow the projection to be solved in polynomial time.

Interior point methods, however, have superlinear asymptotic running times.
In their survey paper, Hegde, Indyk and Schmidt~\cite{DBLP:journals/eatcs/HegdeIS15} therefore
ask whether there is an essentially linear time algorithm that finds a satisfying approximate solution (Open problem 6 in~\cite{DBLP:journals/eatcs/HegdeIS15}).

Formally, they consider two kinds of approximation guarantees.
Suppose $A$ is an algorithm that computes, for each input vector $x$, a $k$-sparse $\Delta$-separated vector $A(x)$.
Since the overall goal is to minimize $||x-A(x)||_2^2$, we impose that $A(x)_i \in \{0,x_i\}$ for all $x$ and all indices $i$.

We say that $A$ has \emph{tail approximation guarantee} $\alpha$
if 
\[
||x-A(x)||_2^2\le \alpha ||x-x^*||_2^2 \text{ for all }x,
\]
where $x^*$ is an optimal solution.
Moreover, we say that $A$ has \emph{head approximation guarantee} $\beta$
if 
\[
||A(x)||_2^2\ge \beta ||x^*||_2^2 \text{ for all }x.
\]
It is easy to see that these two notions are incomparable: one does not imply the other.
As usual, we say that an algorithm is an \emph{$\alpha$-head approximation} (\emph{$\beta$-tail approximation}) if it has head (tail) approximation guarantee $\alpha$ (or $\beta$).

We prove the following.

\begin{theorem}\label{thm:main}
There is an algorithm running in $O(\epsilon^{-2}n)$ time that is both a $(1-\epsilon)$-head approximation and a $(1+\epsilon)$-tail approximation for the $\Delta$-separated sparsity projection problem.
\end{theorem}

We show, furthermore, how our algorithm fits in the  \emph{approximate model iterative hard tresholding} (AM-IHT) framework by Hegde, Indyk and Schmidt~\cite{DBLP:conf/soda/HegdeIS14}. The resulting recovery algorithm for 
 $k$-sparse $\Delta$-separated signals performs significantly faster than the existing algorithm~\cite{HDC09}, 
at least for large input data.

Finally, we make the case that our algorithms are not only fast from a theoretical point of view
but also in practice. To this end we have tested our algorithms  on two types of 
randomly generated data. We present the results in Section~\ref{superscientificsec}.

\medskip
While still a young field, compressed sensing already offers a rich literature. 
We refer Foucard and Rauhut~\cite{DBLP:books/daglib/0036092} for an extensive introduction to the field.
The sparsity model we study here, $\Delta$-separated sparsity, is an example of a whole range of models, 
that fall under the heading of \emph{structured sparsity}. Examples of this come up naturally in 
applications such as image processing, where wavelet 
coefficients of piecewise smooth functions approximately form a subtree of a given rooted tree; 
see~\cite{DBLP:books/daglib/0098272,DBLP:journals/tit/BaraniukCDH10}.

We finally remark that the $\Delta$-separated sparsity projection problem may also be expressed in the language of graphs: indeed, 
an optimal solution corresponds to a maximum weight independent set of size at most $k$ in 
a certain unit interval graph. The maximum weight independent set problem,
without any restriction on the size of the set, admits a linear time algorithm on interval graphs~\cite{Fra76}.
Generalized cardinality restrictions on independent sets have been 
studied too~\cite{DBLP:journals/corr/Bandyapadhyay14,DBLP:conf/caldam/KalraMPP17}.

\section{A combinatorial optimization formulation}

Since $||x||_2^2 = \sum_{i=1}^n x_i^2$, we may study the non-negative vector obtained from the input vector $x$ by taking the square of each entry instead.
This yields the following, more convenient formulation of the $\Delta$-separated sparsity model projection problem.

\begin{hproblem}{Separated Sparsity}
\instance Positive integers $k$, $\Delta$ and a vector $x\in\mathbb R_{\ge 0}^n$ of non-negative reals.
\task Find $I\subseteq [n]$ such that $|I|\leq k$, such that $|i-j|\geq\Delta$
for all distinct $i,j\in I$ and such that $\sum_{i\in I}x_i$ is maximal.
\end{hproblem} 

Obviously, head and tail approximation could as naturally be formulated with respect to other norms. 
At least for $\ell_p$-norms the methods we propose would still work.

We prove the following two statements.

\begin{lemma}\label{thm:head}
There is a $(1-\epsilon)$-head approximation $A$ for \prob{Separated Sparsity} that runs in $O(\epsilon^{-2}n)$ time.
That is, $\sum_{i \in A(x)} x_i \ge (1-\epsilon) \sum_{j \in I^*} x_j$ where $I^*$ is an optimal solution.
\end{lemma}

\begin{lemma}\label{thm:tail}
There is a $(1+\epsilon)$-tail approximation $A$ for \prob{Separated Sparsity} that runs in $\bigO(\epsilon^{-2}n)$ time.
That is, $\sum_{i \notin A(x)} x_i \le (1+\epsilon) \sum_{j \notin I^*} x_j$ where $I^*$ is an optimal solution.
\end{lemma}

By running both of these algorithms and taking the solution with larger value $\sum_{i\in A(x)}x_i$ we obtain an algorithm that is both, a $(1-\epsilon)$-head and a $(1+\epsilon)$-tail approximation. This implies Theorem~\ref{thm:main}, our main result.\footnote{With a little bit of work, one may also prove that the tail approximation algorithm (Algorithm~\ref{tailalgo})
actually achieves the same head approximation guarantee as our head approximation algorithm. This is because we use the head approximation as a subroutine in our tail approximation.}

\section{Preliminaries}

We write $[n]$ for $\{1,\ldots, n\}$.
Let $k$, $\Delta$ be positive integers, let $x \in \mathbb R^n_{\ge 0}$, and
consider an  index subset $I \subseteq [n]$.
We write $x(I)$ for $x(I) = \sum_{i \in I} x_i$, and 
 $\overline I$ for $[n] \setminus I$.
We also write $x_I$ for the vector $z\in\mathbb R^n$ defined by
\[
z_i = \begin{cases}
        x_i & \text{if } i\in I\\
        0 & \text{otherwise }
\end{cases}
\]

Because it will be necessary in the recovery algorithm, we generalize the problem 
setting a bit. For an integer $p$, we say that a set $I\subseteq [n]$ is 
\emph{$p$-spikes $\Delta$-separated} if for every set $Z\subseteq [n]$ of $\Delta-1$
consecutive integers it holds that $|I\cap Z|\leq p$. Any vector $x\in\mathbb R^n$
is \emph{$p$-spikes $\Delta$-separated} if its support (the set of indices with non-zero entries) is $p$-spikes $\Delta$-separated.
Note that $x$ is $1$-spike $\Delta$-separated if and only if $x$ is $\Delta$-separated.

We also define the corresponding projection problem:

\begin{hproblem}{$p$-Spikes Separated Sparsity}
\instance Positive integers $k$, $\Delta$ and a vector $x\in\mathbb R_{\ge 0}^n$ of non-negative reals.
\task Find a $p$-spikes $\Delta$-separated set $I\subseteq [n]$ such that $|I|\leq k$
and such that $x(I)$ is maximal.
\end{hproblem} 

Indeed, we will often even consider a version of \prob{$p$-Spikes Separated Sparsity} that
is restricted to subsets of $[n]$.
For a subset $M$ of $[n]$ and any positive integer $\ell$,
we say that a set $I$ is a \emph{feasible solution} of $\spa^p_x(M,\ell)$ 
if $I$ is a $p$-spikes $\Delta$-separated subset of $M$ of size $|I|\leq\ell$. 
The set $I$ is a solution if $x(I)$ is maximum among
all feasible solutions of  $\spa^p_x(M,\ell)$. Note that whether $I$ is a solution or not also depends on $\Delta$. As $\Delta$ will always be clear from the context
and will not vary within our arguments, we omit it from the notation. We may also drop $x$
if $x$ is clear from the context, and if we simply write $\spa_x(M,\ell)$ rather than 
$\spa^p_x(M,\ell)$ we mean the $1$-spiked version of the problem, i.e., that $p=1$. 

For an optimal solution $I$ of $\spa^p_x(M,\ell)$ we define $\opt^p_x(M,\ell)$
as $x(I)$. Clearly, $\opt^p_x([n],k)$ is the value of the optimal 
solution of \prob{$p$-Spikes Separated Sparsity}.

\medskip
To simplify the presentation we assume that the basic arithmetic operations (addition/subtraction, multiplication/division and comparison) can be performed in a single time step each.
We remark without proof that even if the size of the numbers in the input do play a role, our main algorithms still run in linear time.

\section{Dynamic programming}

Both the normal ($1$-spike) \prob{Separated Sparsity} problem as well as its
$2$-spiked variant can be solved with dynamic programming. We start with one spike.

Let us write 
$\opt(M,\ell)=0$ if $\ell\leq 0$ or if $M=\emptyset$.
We observe that for all $i\in [n]$ and $\ell\in [k]$
\begin{equation}\label{optupdate}
\opt([i],\ell)=\max (x_i+\opt([i-\Delta],\ell-1), \opt([i-1],\ell)).
\end{equation}
A feasible solution of $\spa([i],\ell)$ may be computed at the same time by keeping track 
of whether the maximum in~\eqref{optupdate} is attained by picking $i$ for the feasible solution or not.

We thus obtain the following lemma.

\begin{lemma}\label{dplem}
Given an instance $\spa_x([n],k)$ with $x\in\mathbb R^n_{\geq 0}$, dynamic programming allows to compute 
all values $\opt([n],\ell)$ for  $\ell=1,\ldots,k$, as well as the respective solutions,
in running time $\bigO(kn)$.
\end{lemma}
We note that the algorithm is essentially a special case of the algorithm 
of Bandyapadhyay~\cite{DBLP:journals/corr/Bandyapadhyay14} for the budgeted maximum weight independet
set problem on interval graphs.

As $k$ cannot be considered a constant, but rather might even be linear in~$n$,
the dynamic programming algorithm does not run in linear time. 
Indeed, this is the whole point of this article: to improve on the running time of $\bigO(kn)$.
If~$n$, however, is so large 
that it does not matter, i.e., if $k\geq \tfrac{n}{\Delta}$ then a simpler dynamic
programming approach, that drops the parameter~$\ell$, can be used to compute $\opt([n],k)$ in $\bigO(n)$-time.

We also describe a slightly trickier dynamic programming algorithm to solve
\prob{$2$-Spikes Separated Sparsity}.

\begin{lemma}\label{2dplem}
Given an instance $\Tspa_x([n],k)$ with $x\in\mathbb R^n_{\geq 0}$, dynamic programming allows to compute 
all values $\opt^2_x([n],\ell)$ for $\ell=1,\ldots,k$, as well as the respective solutions, in running time $\bigO(k\Delta n)$.
\end{lemma}
\begin{proof}
To see this, let $\opt^2([r],i,\ell)$ be 
the objective value of an optimal solution $I$ of \prob{2-Spike Separated Sparsity} 
on the vector $x_{[r]}$ such that $|I|\le\ell$ and $|I \setminus [r-i]| \le 1$.
Here, $r \in [n]$, $i \in \{0,\ldots,\Delta-1\}$ and $\ell\in [k]$.
Let us write $\opt^2([r],i,\ell)=0$ if $r\leq 0$ or $\ell\le 0$.
We observe that $\opt^2([r],0,\ell)=\opt^2([r],1,\ell)$ for all $r\in[n]$ and $\ell\in[k]$, and that 
for $i\in \{1,\ldots,\Delta-1\}$ the value $\opt^2([r],i,\ell)$ can be computed as follows:
%formula needs $i\geq 1$ !
\begin{equation*}\label{2optupdate}
\opt^2([r],i,\ell)=\max (x_r+\opt^2([r-i],\Delta-i,\ell-1), \opt^2([r-1],i-1,\ell))
\end{equation*}
This gives rise to an $\bigO(k\Delta n)$ time dynamic programming algorithm that 
outputs the desired values $\opt^2([n],\ell)=\opt^2([n],0,\ell)$, $\ell\in[k]$. The corresponding solutions can be recovered by standard techniques. 
\end{proof}

\section{Head approximation}

In this section we describe a head approximation for \prob{$p$-Spikes Separated Sparsity}.
As a subroutine it needs an exact algorithm which we denote~$A$ below.
For instance, this could our 
a dynamic programming algorithm if $p=1,2$. As the exact algorithm is only 
called for smaller subinstances, the head approximation has a better asymptotic performance.

The main algorithm generates a number of (still large) subinstances that then 
are solved with another method 
 (Algorithm~\ref{slicesolver}) that we will treat afterwards. 

\begin{algorithm}\caption{}\label{headalgo}
\inp An instance $(n,x,\Delta,k)$ of \prob{$p$-Spikes Separated Sparsity} and $\epsilon > 0$.\\
\out A feasible solution.
\begin{algorithmic}[1]
\State Let $\lambda \in \mathbb N$ be the smallest integer such that $\lambda \ge \epsilon^{-1}$.
\For{$\nu=0,\ldots,\lambda$}
\State Set $b=(\lambda+1)\Delta$.
\State Compute the set $S_\nu$ defined as $[n]\sm\bigcup_{j=0}^\infty \{jb+\nu\Delta+1,\ldots,jb+\nu\Delta+\Delta\}$.
\State Solve $\spa^p(S_\nu,k)$ with Algorithm~\ref{slicesolver} optimally, and let $I_\nu$ be the obtained feasible solution.
\EndFor
\State Return $I_{\tau}$ with $x(I_{\tau})=\max_{\nu}x(I_{\nu})$.
\end{algorithmic}
\end{algorithm}

As we will use this twice, once for the head, and once for the 
tail approximation, we prove that at least one of the sets $S_0,\ldots, S_\lambda$
always covers almost all of the weight of any vector $z$.

\begin{lemma}\label{zlem}
For every $z\in\mathbb R^n_{\geq 0}$ there is some $\nu\in\{0,\ldots,\lambda\}$ with 
\[
z(S_\nu)\geq \frac{\lambda}{\lambda+1}z([n]).
\]
\end{lemma}
\begin{proof}
As every $i\in[n]$ appears in exactly $\lambda$ of the sets $S_\nu$ it follows that
\begin{align*}
(\lambda+1)\max_{\nu=0,\ldots,\lambda}z(S_\nu)\geq 
\sum_{\nu=0}^\lambda z(S_\nu) &= \lambda z([n]).
\end{align*}
Choosing $\nu$ as the index that achieves the maximum yields the desired $S_\nu$.
\end{proof}

Building on the above lemma, we can prove the desired approximation guarantee.
We assume here, and we prove it later in Lemma~\ref{lem:slice}, that Algorithm~\ref{slicesolver} works correctly.

\begin{lemma}\label{lem:apx}
Algorithm~\ref{headalgo} is a $(1-\epsilon)$-head approximation for \prob{$p$-Spikes Separated Sparsity}.
\end{lemma}
\begin{proof}
Let $I^*$ be an optimal solution of $\spa^p([n],k)$, and set $x^*=x_{I^*}$.
Thus $x^*([n])=\opt^p([n],k)$. 
Moreover, note that $I^*\cap S_\nu$ is a feasible solution of $\spa^p(S_\nu,k)$ for every $\nu\in\{0,\ldots,\lambda\}$.
Thus $x^*(S_\nu)\leq\opt^p(S_\nu,k)$. With $\nu$ as in Lemma~\ref{zlem} we then obtain
\[
\opt^p(S_\nu,k)\geq x^*(S_\nu)\geq \frac{\lambda}{\lambda+1}x^*([n])
=\frac{\lambda}{\lambda+1}\opt^p([n],k)
\]
This and $\frac{\lambda}{\lambda+1} \ge 1-\epsilon$ finish the proof.
\end{proof}

We remark that the analysis is tight:
take $x=\1$, $k=n$, and $\Delta=1$.
One can also construct examples with larger $\Delta$.

Let us now solve $\spa^p(S_\nu,k)$ optimally.
For a set $S\subseteq [n]$ of integers and an $x\in\mathbb R^n_{\geq 0}$ and integer $\Delta$,
call a subset $B\subseteq [n]$ a \emph{block} of $S$ if $B\cap S$ is non-empty and if $B$
is a set of consecutive integers from $[n]$ that is minimal subject to the property that 
whenever $i\in B$ with $x_i\neq 0$ and $j\in S$ with $x_j\neq 0$ and $|i-j|<\Delta$ then $j\in B$.

Note that the blocks of $S$ are disjoint, and that, 
more strongly, for two distinct blocks $B_1,B_2$ we have that
 $|i-j|\geq\Delta$ for any $i\in B_1$ and $j\in B_2$.
Moreover observe that for every $\nu$ each block of $S_\nu$ as in Algorithm~\ref{headalgo}
has size at most $\lambda\Delta$.

Algorithm~\ref{slicesolver} that we need in order to complete the description of 
our head approximation algorithm calls, as a subroutine, an exact algorithm $A$ 
that solves \prob{$p$-Spikes Separated Sparsity} restricted to each block. 
In fact, we need an algorithm that computes $\opt^p(B,\ell)$ for each $0\leq\ell\leq \left\lceil \frac{|B|}{\Delta} \right\rceil$, where $B$ is a block.
For $p=1,2$, this can be done, quite efficiently, with the dynamic programming algorithms as outlined
in the previous section. However, any algorithm meeting the above requirements could be used in place of $A$.

\begin{algorithm}\caption{}\label{slicesolver}
\inp An instance $\spa^p(S,k)$ with maximal block size at most~$\lambda\Delta$\\
\out An optimal solution of $\spa^p(S,k)$
\begin{algorithmic}[1]
\State Compute the blocks $B_1,\ldots, B_s$ of $S$.
\State For each block $B_t$ compute $\lambda_t=\left\lceil \frac{|B_t|}{\Delta} \right\rceil$ and call algorithm $A$ to 
compute $\opt^p(B_t,\ell)$ for each $\ell\in[\lambda_t]$.
\State For all $t\in [s]$ and all $\ell\in[\lambda_t]$ 
set $q^t_\ell=\opt^p(B_t,\ell)-\opt^p(B_t,\ell-1)$.
\State Compute $Q \subseteq \{(t,\ell) : t \in[s] , \ell \in [\lambda_t]\}$ with $|Q|\leq k$ such that $\sum_{(t,\ell)\in Q}q^t_\ell$ is maximal among all choices of $Q$.
\State Delete each $(t,\ell)$ from $Q$ for which $q^t_\ell=0$. 
\State Return the set $T=\bigcup_{t=1}^s T_t$, where $T_t$ is the optimal solution of $\spa(B_t,j)$ and $j = |Q \cap \{(t,\ell): \ell \in [\lambda_t]\}|$.
\end{algorithmic}
\end{algorithm}

\begin{lemma}\label{lem:slice}
Given an algorithm $A$ that solves instances of \prob{$p$-Spikes Separated Sparsity} optimally,
Algorithm~\ref{slicesolver} solves $\spa^p(S,k)$ optimally.
\end{lemma}
While the lemma might sound a bit tautological at first reading, it actually makes sense:
the point here is not that we can solve  \prob{$p$-Spikes Separated Sparsity}
at all but that we only need to call an exact algorithm for small slices of the whole ground set,
which then leads to a better performance.
\begin{proof}[Proof of Lemma~\ref{lem:slice}]
We start by proving 
\begin{equation}\label{margins}
q^t_\ell\leq q^t_{\ell-1} \emtext{ for all $t,\ell$ with } t\in [s]\emtext{ and }2\leq\ell\leq\lambda_t.
\end{equation}
For this we show, equivalently, that
\begin{equation}
2\cdot\opt^p(B_t,\ell-1)\geq\opt^p(B_t,\ell-2)+\opt^p(B_t,\ell)
\end{equation}
Consider an optimal solution $M$ of $\spa^p(B_t,\ell-2)$ and an optimal 
solution $L$ of $\spa^p(B_t,\ell)$.
Consider the disjoint union $M\cup L$. That is, if some element appears in $M$ and in $L$, we consider
it to appear twice in $M\cup L$. In the natural order, choose for $M'$ every second element of $M\cup L$,
and let $L'$ be the other elements. Obviously, $|M'|=|L'|\le\ell-1$, and neither $M'$ nor $L'$ contains an element twice.

Suppose that $M'$ or $L'$ is not a feasible solution of $\spa^p(B_t,\ell-1)$. 
By symmetry, we may assume this is 
the case for $M'$. Then there must be a set $Z$ of $\Delta$ consecutive integers and 
$p+1$ elements $i_1< \ldots < i_{p+1}$
of $M'$ such that $i_1,\ldots ,i_{p+1}$ all lie in $Z$. 
Since $M'$ consists of every second element of $M\cup L$
there are thus $j_1,\ldots,j_p\in L'$ such that $i_1\leq j_1\leq \ldots\leq j_p\leq i_{p+1}$. 
In particular, all  of these $2p+1$ elements lie in $Z$.  Then $p+1$ of these
must belong to either $M$ or to $L$, %but then all $p+1$ lie in $Z$, 
which contradicts
that $M$ and $L$ are $p$-spikes $\Delta$-separated.
This proves~\eqref{margins}.

Next we prove
\begin{equation}\label{onedir}
\sum_{(t,\ell)\in Q}q^t_\ell \geq \opt^p(S,k).
\end{equation}
For this, consider a feasible solution $I$ of $\spa^p(S,k)$. Consider some $t$, and set $\ell_t=|I\cap B_t|$.
Note that $\ell_t\leq\lambda_t$ since $I$ is feasible.
In turn, since $I\cap B_t$ is a feasible solution of $\spa^p(B_t,\ell_t)$ it follows 
that $x(I\cap B_t)\leq \opt^p(B_t,\ell_t)$. 
We define a set $P\subseteq [s]\times [\lambda]$ with $|P|\leq k$ 
by including $(t,1),\ldots,(t,\ell_t)$ for every $t\in[s]$ in $P$. 
Then 
\begin{align*}
x(I)&=\sum_{t=1}^s x(I\cap B_t) \leq \sum_{t=1}^s \opt^p(B_t,\ell_t)\\
& = \sum_{t=1}^s (q^t_1+\ldots+q^t_{\ell_t}) = \sum_{(t,\ell)\in P}q^t_\ell
\leq \sum_{(t,\ell)\in Q}q^t_\ell.
\end{align*}
This proves~\eqref{onedir}. 

In view of~\eqref{onedir}, the proof of the lemma is finished once we prove that 
$T$ is a  feasible solution of $\spa^p([n],k)$ and that $x(T)=\sum_{(t,\ell)\in Q}q^t_\ell$.
To see that $T$ is a feasible solution, we note that $|T|\leq k$ as $|Q|\leq k$, and that 
$T$ is $p$-spikes $\Delta$-separated, since restricted to each block $B_t$ 
it is $p$-spikes $\Delta$-separated and since
indices from distinct blocks are at least $\Delta$ steps apart.

To determine $x(T)$,
let $t\in [s]$ and $j_t = |Q \cap \{(t,\ell): \ell \in [\lambda_t]\}|$.
By~\eqref{margins} we may assume that $(t,1),\ldots,(t,j_t)\in Q$.
Hence, $(t,m) \notin Q$ for every $m>j_t$.
This implies $x(T_t) =  q^t_1+\ldots+q^t_{j_t}$.
Hence,
\[
x(T) = \sum_{t=1}^s (q^t_1+\ldots+q^t_{j_t}) = \sum_{(t,\ell)\in Q}q^t_\ell,
\]
which completes the proof.
\end{proof}

\begin{lemma}\label{runninglem}
Let $f(r)$ be an upper bound on the running time of algorithm $A$ when run on a vector of dimension $r$.
Then Algorithm~\ref{slicesolver} can be implemented to run in time 
\[
\bigO \left(n \max_{r \in [\lambda \Delta]} 1 + \frac{f(r)}{r} \right).
\]
\end{lemma}
\begin{proof}
For each block $B_t$, Algorithm~\ref{slicesolver} uses algorithm $A$ as a subroutine to compute $\opt(B_t,\ell)$ for $0\leq\ell\leq\lambda_t$.
This can be done in time $f(|B_t|)$.

It remains to discuss the complexity of finding the $k$ largest elements in the vector $P=(q^1_1,\ldots,q^s_{\lambda_s})$.
This can be done in $O(n)$ time using order statistics. 
To see this, note that the vector $P$ is of dimension $\sum_{t=1}^s \lambda_t = O(n)$.

We now find the $k$-th largest element $\hat q$ of $P$ in $O(n)$ time using, for example, the \alg{Introselect} algorithm~\cite{DBLP:journals/spe/Musser97}.
Then we collect the elements of $P$ of value larger than $\hat q$ and put them into our feasible solution.
Finally, we add elements of value equal to $\hat q$ until our feasible solution contains $k$ elements in total.
We obtain a running time of $O(n)$ for this step.
The computation of the $T_t$ and $T$ are straightforward and can be done in $O(n)$ time.

We obtain a total running time of at most
\begin{equation}\label{eqn:RTcomp}
\bigO \left(\max \left\{ \sum_{t=1}^s f(b_t) : s \in \left[ \left\lceil \frac{n}{\Delta} \right\rceil \right], b \in [\lambda \Delta]^s, \sum_{t=1}^s b_t \le n-\Delta(s-1) \right\} +n \right)
\end{equation}
where $b_t$ denotes the size of the $t$-th block $B_t$ and we used the observation that if there are $s$ blocks there are at least $n-\Delta(s-1)$ elements of $[n]$ not contained in a block.
Introducing the running time per element $\frac{f(r)}{r}$ we may simplify~\eqref{eqn:RTcomp} to
\[
\bigO \left(n \max_{r \in [\lambda \Delta]} 1 + \frac{f(r)}{r} \right),
\]
and hence the proof is complete.
\end{proof}

We can now prove Lemma~\ref{thm:head}.

%\begin{lemma}\label{lem:RT}
%For \prob{$1$-Spike Separated Sparsity}, Algorithm~\ref{headalgo}
% is a $(1-\epsilon)$-head approximation that 
%can be implemented to run in $O(\epsilon^{-2} n)$ time.
%\end{lemma}
\begin{proof}[Proof of Lemma~\ref{thm:head}]
Lemma~\ref{lem:apx} says that Algorithm~\ref{headalgo} is a $(1-\epsilon)$-approximation.

To compute the running time, first observe that, by Lemma~\ref{dplem}, dynamic 
programming in place of algorithm~$A$ has a running time of $f(r) = \bigO(r \cdot \frac{r}{\Delta})$,
and thus Lemma~\ref{runninglem} yields a running time of
\[
\bigO \left(n \max_{r \in [\lambda \Delta]} 1 + \frac{f(r)}{r} \right) = \bigO \left(n \max_{r \in [\lambda \Delta]} 1 + \frac{r}{\Delta} \right) =  \bigO (\lambda n)
\]
for Algorithm~\ref{slicesolver}.

Consequently, we obtain a running time of $\bigO(\lambda^2 n)$ for Algorithm~\ref{headalgo}.
Then, $\bigO(\lambda^2 n)=\bigO(\epsilon^{-2}n)$, by the choice of $\lambda$,
 and thus the proof is complete.
\end{proof}

\begin{theorem}\label{thm:2-head}
For \prob{$2$-Spikes Separated Sparsity}, Algorithm~\ref{headalgo} is a $(1-\epsilon)$-head approximation that 
can be implemented to run in $\bigO(\epsilon^{-2} \Delta n)$ time.
%That is, $\sum_{i \in A(x)} x_i \ge (1-\epsilon) \sum_{j \in I^*} x_j$ where $I^*$ is an optimal solution.
\end{theorem}
\begin{proof}
The proof is the same as that of Lemma~\ref{thm:head},
with the single exception that we use Lemma~\ref{2dplem} rather than Lemma~\ref{dplem}.
\end{proof}

Algorithm~\ref{headalgo} can also speed up the solution of 
\prob{$p$-Spikes Separated Sparsity} for larger $p$, provided we have access to some exact
algorithm. Especially for large $p$, the dynamic programming approach 
does not seem to be feasible anymore when $n$ grows large. We may, however, encode the problem
as an integer program with a totally unimodular constraint matrix, as has been demonstrated by Hegde et 
al.~\cite{HDC09}, and then solve the linear relaxation. 
The resulting speed-up can be directly read off of Lemma~\ref{runninglem}.

\section{Tail approximation}

In this section we restrict our attention to the 1-spike case exclusively.
Recall that an algorithm $A$ for \prob{Separated Sparsity} has tail approximation guarantee $\alpha$
if 
\[
x(\overline{A(x)})\le \alpha x(\overline{I^*})\text{ for all }x,
\]
where $I^*$ is an optimal solution.

Algorithm~\ref{headalgo} has no constant tail approximation guarantee. 
To see that, note that if $x$ is the all-ones vector, $k=n$, and $\Delta = 1$, the optimal solution $I^*$ is $[n]$.
Consequently, $x(\overline{I^*}) = 0$ and thus Algorithm~\ref{headalgo} would need to solve
the instance optimally to  have a constant tail approximation guarantee. This, however, 
is not the case as $[n]\nsubseteq S_\nu$ for all~$\nu$ provided that $n$ is sufficiently large.
%To see that, choose  $\lambda<k$ and 
%$n>(\lambda+1)^2\Delta$. Put $b=(\lambda+1)\Delta$ and set $I=\{\nu(b+\Delta)+1:\nu=0,\ldots,\lambda\}$.
%Then $I$ is $\Delta$-separated and, moreover, $I\nsubseteq S_\nu$ for all $\nu$. We now define $x\in\mathbb R^n$
%by setting $x_i=1$ if $i\in I$ and $x_i=0$ otherwise. 
%Then $x=x_{I^*}$, where $I^*$ is the optimal solution, 
 %and thus Algorithm~\ref{headalgo} would need to solve
%the instance optimally to  have a constant tail approximation guarantee. This, however, 
%is not the case as $I\nsubseteq S_\nu$ for all~$\nu$. 

\medskip
There is, however, a very simple algorithm that has a constant tail approximation guarantee:
simply pick the best feasible solution among the $k$ largest elements of $x$ (Algorithm~\ref{stupidalgo}).
We briefly discuss the algorithm because it is so simple.

\begin{algorithm}\caption{}\label{stupidalgo}
\inp An instance $\spa([n],k)$\\
\out A feasible solution to $\spa([n],k)$
\begin{algorithmic}[1]
\State Compute index set $L$ of $k$ largest elements in $x$.
\State Solve $\spa(L,k)$ optimally, and output the solution $I$.
\end{algorithmic}
\end{algorithm}

\begin{proposition}
Algorithm~\ref{stupidalgo} has tail approximation guarantee~$2$.
\end{proposition}
We note that, with a bit of care, the algorithm can be implemented to run in $\bigO(n)$ time.
\begin{proof}
Let $I^*$ be an optimal solution, and set $x_{I^*}=x^*$. Then $x(I^*)\leq x(L)$ and thus
\begin{equation}\label{stupid}
x(\overline{I^*})\geq x(\overline L).
\end{equation}
Moreover, $x(I^*\cap L)\leq x(I)$ since $I$ is an optimal solution of $\spa(L,k)$,
while $I^*\cap L$ is some feasible solution of $\spa(L,k)$. Thus $x(L\sm (I^*\cap L))\geq x(L\sm I)$
and 
\begin{equation}\label{stupidtoo}
x(\overline{I^*})\geq x(L\sm (I^*\cap L))\geq x(L\sm I).
\end{equation}
As $x(\overline I)=x(\overline L) + x(L\sm I)$ we obtain with~\eqref{stupid} and~\eqref{stupidtoo}
that $x(\overline I)\leq 2x(\overline{I^*})$.
\comment{
Recall that we can compute the index set $L$ of the $k$ largest elements in $O(n)$ time using the \alg{Introselect} algorithm~\cite{DBLP:journals/spe/Musser97}.
To solve $\spa(L,k)$ optimally, we run dynamic programming on the vector $y=x_L$ to compute $\opt_y([n],n)$.
According to Lemma~\ref{lem:dynproglargek}, we can compute and output the feasible solution in $O(n)$ time.
After removing all elements not in $L$, we obtain the optimal solution $I$ of $\spa(L,k)$.}%
\end{proof}

Unfortunately, the analysis is tight. Indeed, consider the instance $x=(1,1,1)$, $\Delta=2$ und $k=2$. 
Then the algorithm might find $L=\{1,2\}$, 
which might result in $I=\{1\}$. Then $x(\overline I)=2$ but $I^*=\{1,3\}$ and $x(\overline{I^*})=1$.

\medskip
We now come to our tail approximation of arbitrary precision.
Before stating the algorithm, we need two technical lemmas.

\begin{lemma}\label{lem:am}
Let $n$ be a positive integer, $W\subseteq [n]$, $A , B\subseteq W$, and $x\in\mathbb R^n_{\geq 0}$.
Assume that
\[
x(A)\geq \alpha x(B) \text{ and }x(B) \leq \mu x(W)
\]
for some reals $0<\alpha,\mu<1$. Then 
\[
x(W\sm A) \le \frac{\alpha}{1-\frac{1-\alpha}{1-\mu\alpha}}\cdot x(W\sm B).
\]
\end{lemma}
\begin{proof}
Put $\gamma=\frac{1-\alpha}{1-\mu\alpha}$ and observe that $0<\gamma<1$.
Now
\begin{align*}
x(W\sm  B) & = x(W)-(1-\gamma)x(B) -\gamma x(B) \\
&\geq x(W) -\frac{1-\gamma}{\alpha}x(A) - \gamma\mu x(W)  \\
& = \frac{1-\gamma}{\alpha}x(W\sm  A) + \left(1-\frac{1-\gamma}{\alpha}\right) x(W) - \gamma\mu x(W)\\
&  = \frac{1-\gamma}{\alpha}x(W\sm  A),
\end{align*}
since
\[
1-\frac{1-\gamma}{\alpha}-\gamma\mu = \frac{1}{\alpha}(\alpha-1+\gamma-\alpha\mu\gamma)
=\frac{1}{\alpha}(\alpha-1+\gamma (1-\mu\alpha)) = 0.
\]
\end{proof}

Given a vector $x \in \mathbb R^n$ we define the \emph{tail vector} $t \in \mathbb R^n$ of $x$ by 
setting
\[
t_i = x(\{j \in [n] : j \neq i, |i-j| < \Delta\}) = \sum_{j=i-\Delta+1}^{i+\Delta-1}x_j-x_i
\]
for all $i \in [n]$.
We call an index $i \in [n]$ \emph{strong} if $x_i > t_i$ and \emph{weak} otherwise.
It is clear that for any two strong indices $i,j$ we have $|i-j|\geq \Delta$.

For a given $x \in \mathbb R^n$ let $S$ be the set of its strong indices.
We define the \emph{reduced vector} $r \in \mathbb R^n$ of $x$ by setting
\[
r_i = \begin{cases}
				x_i & \text{if } i \in S \mbox{ or } |i-j| \ge \Delta \text{ for all } j \in S\\
        0 & \text{otherwise }
\end{cases}
\]
for all $i\in [n]$.

\begin{lemma}\label{redlem}
Consider an instance $\spa_x([n],k)$, let $S$ be the set of strong
indices of $x$,
and let $r$ be the reduced vector of $x$. Let $I^*$ be an optimal solution of $\spa_x([n],k)$.
If $i\in I^*\sm S$ and $s\in S$ then $|s-i|\geq\Delta$. 
In particular, $\opt([n],k) = x(I^*) = r(I^*)$.
\end{lemma}
\begin{proof}
Suppose that there are  $i\in I^*\sm S$ and $s\in S$ with $|s-i|<\Delta$. 
Set $C=\{j\in [n]: |s-j|< \Delta,\,s\neq j\}$, and consider $I=(I^*\sm C)\cup\{s\}$. 
Then $I$ is a feasible solution of $\spa([n],k)$, and
 as $s$ is strong, it holds that $x(C\cap I^*)\leq x(C)<x_s$, which implies $x(I)>x(I^*)$.
Obviously, this is impossible because $I^*$ is an optimal solution.
\end{proof}

We now present our tail approximation for \prob{Separated Sparsity}.
For this, fix $\epsilon >0$.

\begin{algorithm}\caption{}\label{tailalgo}
\inp An instance $(n,x,\Delta,k)$ of \prob{Separated Sparsity} and $\epsilon > 0$\\
\out A feasible solution to $\spa([n],k)$
\begin{algorithmic}[1]
\State Compute the tail vector $t$ and the set $S$ of all strong indices.
\State Compute the reduced vector $r$ of $x$.
\State Run Algorithm~\ref{headalgo} (for $p=1$ spike)
on $r$ with precision $1-\frac{\epsilon}{2}$, subject to the following modification: 
Instead of calling Algorithm~\ref{slicesolver} on $S_\nu$, call it on $S_\nu \cup S$ for $\nu=0,\ldots,\lambda$. 
\end{algorithmic}
\end{algorithm}

\begin{lemma}
Algorithm~\ref{tailalgo} returns a $(1+\epsilon)$-tail approximation of the \prob{Separated Sparsity} problem.
\end{lemma}
\begin{proof}
First note that the call of Algorithm~\ref{slicesolver} is valid, as the maximal block size
of $S\cup S_\nu$ with respect to $r$ is at most $\lambda\Delta$: indeed, the elements of $S$ form singleton blocks,
while every other block is contained in a block of $S_\nu$ with respect to $x$, and thus has block size 
at most $\lambda\Delta$.
 
Let $I^*$ be an optimal solution, and let $W=\overline S$ be the set of weak indices of~$x$. 
Applying Lemma~\ref{zlem} to $z=x_{I^*\cap W}$, we find a $\nu$ with
\begin{equation}\label{nuweak}
x(S_\nu\cap I^*\cap W)=z(S_\nu) \ge (1-\tfrac{\epsilon}{2}) z([n])= (1-\tfrac{\epsilon}{2}) x(I^*\cap W).
\end{equation}

Consider the feasible solution $L=(S \cup S_\nu) \cap I^*$ of $\spa_r(S\cup S_\nu,k)$.
Then
\begin{equation}\label{eqn:tailswins}
\opt_x(S\cup S_\nu,k) \geq \opt_r(S\cup S_\nu,k) \ge r(L) = x(L),
\end{equation}
where the last equality is due to Lemma~\ref{redlem}.

Next, we claim that
\begin{equation}\label{lem:weaktails}
x(W\setminus L) \le (1+\epsilon) x(W\setminus  I^*).
\end{equation}
Before we prove the claim, we note how to finish the proof of the lemma with it. 
For this, observe that 
the choice of $L$ implies $x(S\setminus L) = x(S\setminus I^*)$.
Then 
\begin{align*}
x(\overline L) &= x(S\setminus L) + x(W\setminus L)\\
& \le x(S\setminus I^*) + (1+ \epsilon)x(W\setminus I^*)\\
& \le (1+\epsilon) x(\overline{I^*}).
\end{align*}
If $I$ is the output of Algorithm~\ref{tailalgo} then this implies
\begin{align*}
x(\overline I) \leq x([n])-\opt_x(S\cup S_\nu,k)\overset{\eqref{eqn:tailswins}}{\leq} 
x([n])-x(L) =x(\overline L)\leq (1+\epsilon) x(\overline{I^*}),
\end{align*}
which is the statement of the lemma.

\medskip
Let us prove~\eqref{lem:weaktails}. 
As a consequence of Lemma~\ref{redlem}, we get%\mymargin{sportlich?}
\[
2x(W\sm I^*)\geq t(I^*\cap W).
\]
This, in turn, implies 
\begin{align*}
x(W) &= x(I^* \cap W) + x(W \setminus I^*)\\
& \ge x(I^* \cap W) + \frac{1}{2} t(I^* \cap W)\\
& \ge x(I^* \cap W) + \frac{1}{2} x(I^* \cap W) = \frac{3}{2} x(I^* \cap W),
\end{align*}
where the second inequality is because of what it means for an index to be weak.

Set $A=S_\nu\cap I^*\cap  W$ and $B=I^*\cap W$. 
We may apply Lemma~\ref{lem:am} with $\mu=\frac{2}{3}$, by the preceding inequality, and 
 $\alpha=1-\frac{\epsilon}{2}$, by~\eqref{nuweak}, and obtain
\begin{align*}
x(W\sm L) = x(W\sm A) \leq \frac{\alpha}{1-\frac{1-\alpha}{1-\mu\alpha}} \cdot x(W\sm B)
 =  \frac{\alpha}{1-\frac{1-\alpha}{1-\mu\alpha}}\cdot x(W\sm I^*).
\end{align*}
As
\[
\frac{\alpha}{1-\frac{1-\alpha}{1-\mu\alpha}}
=\frac{1-\frac{\epsilon}{2}}{1-\frac{\frac{\epsilon}{2}}{1-\frac{2}{3}(1-\frac{\epsilon}{2})}} 
= \frac{1-\frac{\epsilon}{2}}{1-\frac{3 \epsilon}{2+2\epsilon}}
= \frac{1-\frac{\epsilon}{2}}{\frac{2-\epsilon}{2+2\epsilon}}
= \frac{2+\epsilon-\epsilon^2}{2-\epsilon}
= 1+ \epsilon,
\]
we have proved~\eqref{lem:weaktails} and thus the lemma.
\end{proof}

It remains to observe that the algorithm runs in linear time.

\begin{lemma}
Algorithm~\ref{tailalgo} can be implemented to run in $\bigO(\epsilon^{-2}n)$ time.
\end{lemma}
\begin{proof}
First note that we can compute $t$ in $\bigO(n)$ time using dynamic programming since 
\[
t_i  = t_{i-1} + x_{i-1} - x_i + x_{i+\Delta-1} - x_{i-\Delta},
\] 
where we assume that all undefined values are~$0$.
From $t$ we can clearly compute $S$ and then $r$ in $\bigO(n)$ time.

Note that the maximum block size of $S_\nu \cup S$ considered by Algorithm~\ref{slicesolver} is $\lambda \Delta$, as the definition of $r$ and the blocks implies that every element of $S$ is placed in a private block.
Hence, we may apply Lemma~\ref{runninglem} as in the proof of Lemma~\ref{thm:head} to complete the proof.
\end{proof}

The above two lemmas prove Lemma~\ref{thm:tail}.

\section{Recovering $k$-sparse $\Delta$-separated signals}

The recovery problem for $k$-sparse $\Delta$-separated signals $x\in \mathbb R^n$ is as follows.
For a sensing matrix $A$, we are given the set of noisy measurements, a vector $y \in \mathbb R^m$,
which relates to $x$ via $y = Ax + e$ for some noise vector $e \in \mathbb R^m$.
The task consists in recovering the original signal $x$, or a vector close to~$x$. 
Hegde, Duarte and Cevher~\cite{HDC09} have shown that this is possible if $A$ is an i.i.d.~subgaussian matrix and $m \ge C \log(\frac{n}{k}-\Delta)$ for some constant $C$.

To do so, Hegde et al.\ prove that such a matrix enjoys some form of the \emph{restricted isometry property} (RIP).
In general, a matrix $A \in \mathbb R^{m \times n}$ is said to have the RIP with constant $\delta$ if every $k$-sparse vector $x \in \mathbb R^n$ satisfies
\begin{equation}\label{eqn:RIP}
(1-\delta) ||x||_2^2 \le ||A x||_2^2 \le (1+\delta) ||x||_2^2.
\end{equation}
In our context, it is enough if \eqref{eqn:RIP} is satisfied for all $k$-sparse $\Delta$-separated vectors $x$.
To formalize this, let us say that a matrix $A$ has the $(k,\Delta,p)$-RIP with constant $\delta$ if \eqref{eqn:RIP} holds for every $k$-sparse $p$-spike $\Delta$-separated vector $x \in \mathbb R^n$.

Hegde, Duarte and Cevher prove the following bound for the $(k,\Delta,1)$-RIP in Theorem~\ref{thm:RIP} and explain how it carries over to the $(k,\Delta,2)$-RIP.
Using the same reasoning, it does in fact extend to the $(k,\Delta,p)$-RIP for every fixed $p$.

\begin{theorem}[Hegde, Duarte and Cevher~\cite{HDC09}]\label{thm:RIP}
Fix $p \in \mathbb N$.
There is a constant $C$ such that, for all $\delta >0$, any $t > 0$ and any
\begin{equation}\label{eqn:p-rip}
m \ge C\left(\delta^{-2} (k \log \left(\frac{n}{k} - \Delta \right) + t -k \ln \delta) \right)
\end{equation}
an $m \times n$ i.i.d.~subgaussian random matrix has the $(k,\Delta,p)$-RIP
with constant $\delta$ with probability at least $1-\epsilon^{-t}$.
\end{theorem}

Using the \emph{Model-based CoSaMp} framework, the authors develop a recovery algorithm using $(k,\Delta,2)$-RIP matrices and linear programming to solve both \prob{Separated Sparsity} and \prob{2-Spike Separated Sparsity} exactly.
Indeed, they give an integer programming formulation of the above-mentioned problems and use total unimodularity of the restriction matrix to argue that solving the linear programming formulation are sufficient.
They obtain an algorithm that converges geometrically in the sense that, after $O(\log \frac{||x||_2}{||e||_2})$ iterations, a $k$-sparse $\Delta$-separated vector $\hat x$ is found with 
\begin{equation}\label{eqn:recovered}
||x-\hat x||_2 \le C ||e||_2 \text{ for some constant } C,
\end{equation}
where $x$ is the vector to be recovered.
The bound~\eqref{eqn:recovered} is generally considered to be a successful recovery of $x$.

If one is bound to use head and tail approximations rather than exact projection algorithms one can still recover the original signal efficiently.
We follow here the \emph{approximate model iterative hard thresholding} (AM-IHT) method, 
a framework proposed by Hegde, Indyk and Schmidt~\cite{DBLP:conf/soda/HegdeIS14}.
It resembles the iterative hard thresholding algorithm by Blumensath and Davies~\cite{BLUMENSATH2009265}.

In the statement of the algorithm below, $T$ denotes the tail approximation according to Theorem~\ref{thm:main} and $H$ denotes the head approximation for the 2-spike case (Theorem~\ref{thm:2-head}). 
We assume that both algorithms run with a fixed precision to be determined later.

\begin{algorithm}\caption{}\label{alg:am-iht}
\inp $A$, $y$, $e$ with $Ax=y+e$ and $i \in \mathbb N$\\
\out $\text{AM-IHT}(y,A,i)$
\begin{algorithmic}[1]
\State $x^0 \gets 0$
\For{$j$ from $0$ to $i-1$} 
\State $x^{j+1} \gets T(x^j + H(A^T(y-Ax^j)))$
\EndFor
\State Return $\text{AM-IHT}(y,A,i) \gets x^{i}$
\end{algorithmic}
\end{algorithm}

As Hegde et al.\ prove,
AM-IHT is a rapidly converging recovery algorithm if the approximation algorithms used run with sufficient precision.
To keep it simple, we do not state their theorem in full generality, but rather customize it for the $\Delta$-separated sparsity model.

\begin{theorem}[Hegde, Indyk and Schmidt~\cite{DBLP:conf/soda/HegdeIS14}]\label{thm:AM-IHT}
Assume that $A$ has the $(\Delta,k,4)$-RIP with constant $\delta$, that $c_T$ is the approximation guarantee of our tail approximation for the 1-spike case and that $c_H$ is the guarantee of our head approximation for the 2-spike case.
Writing $r^j=x-x^j$, we have
\begin{align*}
||r^{i+1}||_2 \le & (1+c_T) \left(  \frac{\sqrt{1-c_H^2}(1+\delta)+\delta}{c_H} + 2\delta  \right) ||r^i||_2\\
&+ (1+c_T) \sqrt{1+\delta} \left(  \frac{\sqrt{1-c_H^2}+1}{c_H} + 4 \right) ||e||_2.
\end{align*}
\end{theorem}

Theorem~\ref{thm:AM-IHT} warrants geometric convergence if $c_T$ and $c_H$ are close enough to 1 and $\delta$ is sufficiently small.
For example, if we run the head approximation for the 2-spike case and the tail approximation for the 1-spike case with $\epsilon = 0.01$ as well as Theorem~\ref{thm:RIP} with $\delta = 0.01$, Theorem~\ref{thm:AM-IHT} yields
\begin{align*}
||x-\text{AM-IHT}(y,A,i)||_2 \le 0.35^i ||x||_2 + 16 ||e||_2,
\end{align*}
where $\text{AM-IHT}(y,A,i)$ denotes the vector obtained after $i$ iterations of the AM-IHT procedure.
Consequently, after $O(\log \frac{||x||_2}{||e||_2})$ iterations we arrive at a $k$-sparse $\Delta$-separated vector $\hat x$ satisfying~\eqref{eqn:recovered}, as desired.

Let $\Lambda$ be the time needed to multiply a vector with an i.i.d.~subgaussian $m\times n$ matrix where $m$ is given by~\eqref{eqn:p-rip}.
Note that $\Lambda = O(n k \log{(\frac{n}{k}-\Delta)})$.
For a fixed precision, the tail approximation runs in $O(n)$ time according to Theorem~\ref{thm:main} and the head approximation runs in $O(n\Delta)$ time according to Theorem~\ref{thm:2-head}.
In view of Algorithm~\ref{alg:am-iht}, we obtain the following.

\begin{corollary}
The AM-IHT algorithm can be implemented to run in $O(n\Delta+\Lambda)$ time steps per iteration.
\end{corollary}

This beats the algorithm by Hegde, Duarte and Cevher~\cite{HDC09} which needs $O(n^{3.5}L^2+\Lambda)$ time per iteration where $L$ is the number of bits of the input.
The bottleneck in the running time of their algorithm is to solve the linear program formulations of \prob{Separated Sparsity} and \prob{2-Spike Separated Sparsity} exactly using an interior point method.
If, instead, we plug in the dynamic programs according to Lemma~\ref{dplem} and~\ref{2dplem}, 
the algorithm needs $O(nk\Delta+\Lambda)$ time per iteration.
We remark that their algorithm was not necessarily designed to obtain a fast asymptotic running time and that it might perform much better compared to ours on smaller instances.

\section{Experiments}\label{superscientificsec}

We have implemented the main algorithms of this paper and investigated running times 
and solution qualities for different random instances. The code,
which is available online\footnote{available at {\small\texttt{http://www.uni-ulm.de/bruhn}}},
 is written in Python, 
and runs on a 
standard desktop computer. We have deviated in one
point from the algorithms as described here: to pick the $k$ largest elements
in Algorithm~\ref{slicesolver} we have used a heap-based 
library method of Python rather than \alg{Introselect}.

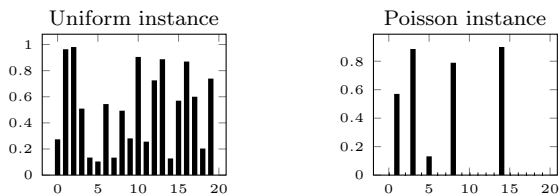
\begin{figure}[!htb]
\centering
\begin{tikzpicture}
\begin{axis}[
title={Uniform instance},
 no markers,
 ycomb,
tiny,
ymin=0,
]
\addplot[ line width = 2pt] table {uniforminst.dat};
\end{axis}
\end{tikzpicture}
\hspace{1cm}
\begin{tikzpicture}
\begin{axis}[
title={Poisson instance},
 no markers,
 ycomb,
tiny,
ymin=0,
]
\addplot[ line width = 2pt] table {purepoisson.dat};
\end{axis}
\end{tikzpicture}

\caption{Examples of a uniform and a Poisson process instance with expected arrival time~$4$}\label{instfig}
\end{figure}

We have tested the algorithms  on two types of random instances. In a \emph{uniform instance}, 
we have picked $x\in\mathbb R^n$ in such a way that each entry $x_i$ receives uniformly at random
a value in the interval $[0,1)$. The second type are \emph{Poisson instances}, where a Poisson
process is run to identify a number of positions $I\subseteq[n]$, the \emph{spikes} of the signal, 
where $x_i$, $i\in I$, receives a value chosen uniformly at random from the interval $[0,1)$;
all entries $x_j$ outside $I$ have value $0$. The expected arrival time of the process, i.e., the 
expected difference between consecutive spikes, is set to the same $\Delta$ as in the instance
$\spa_x([n],k)$, or to $\tfrac{1}{2}\Delta$ when we treat \prob{$2$-Spikes Separated Sparsity}. 
Two example instances are shown in Figure~\ref{instfig}.
Each data point in the Figures~\ref{runningfig}--~\ref{2qualifig} 
is the mean over~$100$ runs with the same parameters.

\def\dyp{*}
\def\hsa{asterisk}
\def\hsb{triangle*}
\def\hsc{square*}
\def\tsa{10-pointed star}
\def\tsb{diamond*}
\def\tsc{pentagon*}

\comment{
\def\dypcol{blue}
\def\hsacol{orange}
\def\hsbcol{red}
\def\hsccol{magenta}
\def\tsacol{cyan}
\def\tsbcol{brown}
\def\tsccol{purple}
}

\def\dypcol{darkgray}
\def\hsacol{darkgray}
\def\hsbcol{darkgray}
\def\hsccol{darkgray}
\def\tsacol{darkgray}
\def\tsbcol{darkgray}
\def\tsccol{darkgray}

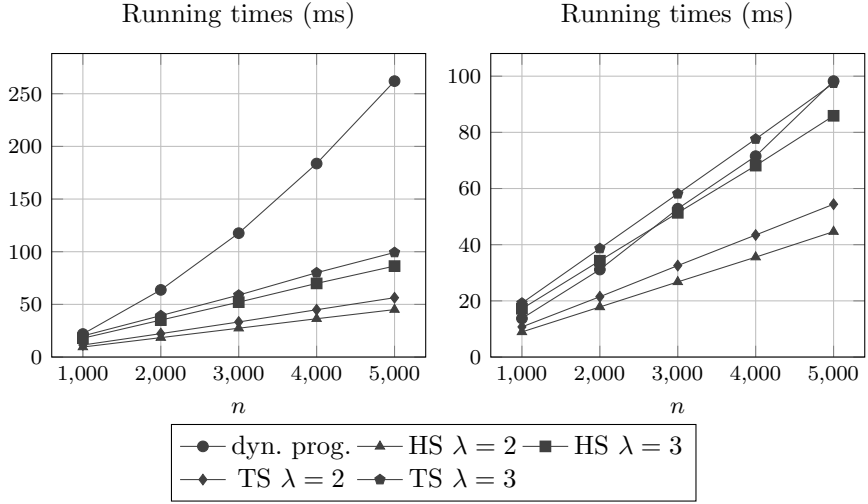
\begin{figure}[!htb]
\centering
\begin{tikzpicture}
\begin{axis}[
title={Running times (ms)},
xlabel={$n$},
grid=major,
cycle list={
  {\dypcol,mark=\dyp},
  {\hsbcol,mark=\hsb},
  {\hsccol,mark=\hsc},
  {\tsbcol,mark=\tsb},
  {\tsccol,mark=\tsc}
},
legend entries={dyn.\ prog.,HS $\lambda=2$,HS $\lambda=3$, TS $\lambda=2$,TS $\lambda=3$},
legend columns=3,
legend to name=named,
small,
ymin=0,
]
\addplot table {timeresA.dat};
\addplot table {timeresB.dat};
\addplot table {timeresC.dat};
\addplot table {timeresD.dat};
\addplot table {timeresE.dat};
\end{axis}
\end{tikzpicture}
\begin{tikzpicture}
\begin{axis}[
title={Running times (ms)},
xlabel={$n$},
cycle list={
  {\dypcol,mark=\dyp},
  {\hsbcol,mark=\hsb},
  {\hsccol,mark=\hsc},
  {\tsbcol,mark=\tsb},
  {\tsccol,mark=\tsc}
},
grid=major,
small,
ymin=0,
]
\addplot table {logtimeA.dat};
\addplot table {logtimeB.dat};
\addplot table {logtimeC.dat};
\addplot table {logtimeD.dat};
\addplot table {logtimeE.dat};
\end{axis}
\end{tikzpicture}
\\
\ref{named}
\caption{Comparison of running times of dynamic programming and Algorithm~\ref{headalgo} (HS) and 
Algorithm~\ref{tailalgo} (TS),
means of 100 repeats per data point. 
Left: $\Delta=\lfloor\tfrac{1}{2}\sqrt n\rfloor=k$. Right: $\Delta=40$, $k=\lfloor\log_2(n)\rfloor$.}\label{runningfig}
\end{figure}

In Figure~\ref{runningfig}, we have plotted the running times of the dynamic programming algorithm,
as well as Algorithms~\ref{headalgo} and Algorithm~\ref{tailalgo} with $\lambda=2,3$ each. 
As the actual instance should not have a significant impact on the running times, 
we have only tested uniformly generated instances. The left diagram shows 
the running times for varying $n$ and $\Delta=k=\lfloor\tfrac{1}{2}\sqrt n\rfloor$. Clearly, 
Algorithms~\ref{headalgo} and Algorithm~\ref{tailalgo} perform much better than the dynamic programming algorithm. 
This is not so pronounced in the right diagram, where we have fixed $\Delta$ to~$40$, while $k$
is set to $\lfloor \log_2 n\rfloor$. As the asymptotic running time of dynamic programming is $\bigO(kn)$,
it is not at all surprising that dynamic programming fares better when $k$ is small.

\begin{figure}[!htb]
\centering
\begin{tikzpicture}
\begin{axis}[
title={Head approximation in per cent},
xlabel={$k$},
cycle list={
  {\hsacol,mark=\hsa},
  {\hsbcol,mark=\hsb},
  {\hsccol,mark=\hsc},
},
grid=major,
legend entries={HS $\lambda =1$, HS $\lambda=2$,HS $\lambda=3$},
legend columns=3,
legend to name=framed,
small,
]
\addplot table {uniheadA.dat};
\addplot table {uniheadB.dat};
\addplot table {uniheadC.dat};
\end{axis}
\end{tikzpicture}
\begin{tikzpicture}
\begin{axis}[
title={Tail approximation  in per cent},
xlabel={$k$},
cycle list={
  {\hsacol,mark=\hsa},
  {\hsbcol,mark=\hsb},
  {\hsccol,mark=\hsc},
},
grid=major,
small,
]
\addplot table {unitailA.dat};
\addplot table {unitailB.dat};
\addplot table {unitailC.dat};
\end{axis}
\end{tikzpicture}
\\
\ref{framed}
\caption{Head and tail approximation in per cent, means of 100 repeats per data point.
Uniform instance with $\Delta=20$ and $n=1000$; only Algorithm~\ref{headalgo} (HS)
shown as the results for Algorithm~\ref{tailalgo} are nearly identical}\label{uniqualifig}
\end{figure}
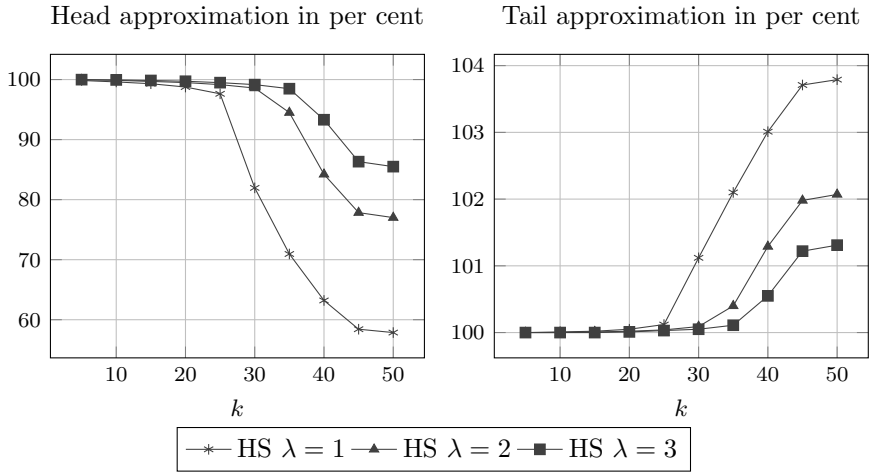

Figure~\ref{uniqualifig} shows  head and tail approximation guarantees of Algorithm~\ref{headalgo}
for different $\lambda$, when applied to uniform instances of size $n=1000$. 
As arguably, the performance 
does not depend that much on the actual value of $\Delta$ compared to $n$ but rather 
on $k\Delta$ compared to $n$, we have kept $\Delta$ fixed to~$20$ and only varied $k$.
In a uniform instance the expected number of strong indices is vanishingly small. Clearly, 
if there are no strong indices then Algorithm~\ref{tailalgo} defaults to Algorithm~\ref{headalgo},
and indeed our test showed no difference between the two in this setting. Consequently, we
have  omitted the results of Algorithm~\ref{tailalgo} in Figure~\ref{uniqualifig}. 
As can be seen, the algorithm  exhibits very good tail approximation 
even for $\lambda=1$ and reasonably good head approximation at least for $\lambda=2,3$. 

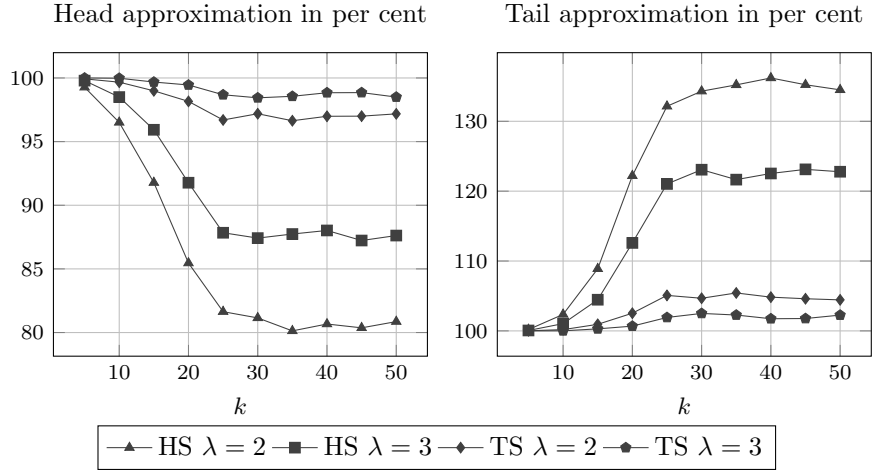
\begin{figure}[!htb]
\centering
\begin{tikzpicture}
\begin{axis}[
title={Head approximation in per cent},
xlabel={$k$},
cycle list={
  {\hsbcol,mark=\hsb},
  {\hsccol,mark=\hsc},
  {\tsbcol,mark=\tsb},
  {\tsccol,mark=\tsc}
},
grid=major,
legend entries={HS $\lambda=2$,HS $\lambda=3$, TS $\lambda=2$,TS $\lambda=3$},
legend columns=4,
legend to name=blamed,
small,
]
\addplot table {pheadA.dat};
\addplot table {pheadB.dat};
\addplot table {pheadC.dat};
\addplot table {pheadD.dat};
\end{axis}
\end{tikzpicture}
\begin{tikzpicture}
\begin{axis}[
title={Tail approximation  in per cent},
xlabel={$k$},
cycle list={
  {\hsbcol,mark=\hsb},
  {\hsccol,mark=\hsc},
  {\tsbcol,mark=\tsb},
  {\tsccol,mark=\tsc}
},
grid=major,
small,
]
\addplot table {ptailA.dat};
\addplot table {ptailB.dat};
\addplot table {ptailC.dat};
\addplot table {ptailD.dat};
\end{axis}
\end{tikzpicture}
\\
\ref{blamed}
\caption{Head and tail approximation in per cent, means of 100 repeats per data point.
Poisson process instance with expected arrival time $\Delta=20$ and $n=1000$.}\label{poissonfig}
\end{figure}

Algorithms~\ref{headalgo} and~\ref{tailalgo} only differ when there are a (substantial) number
of strong indices. Obviously, Poisson instances are set up precisely in this way as to 
yield a good number of strong indices. Not surprisingly, Algorithm~\ref{tailalgo}
handles these instances better than Algorithm~\ref{headalgo}, as can be seen in Figure~\ref{poissonfig}. 
Algorithm~\ref{tailalgo} shows very good head and tail approximation, even for $\lambda=2$, while 
both are much poorer for Algorithm~\ref{headalgo}. 

\begin{figure}[!htb]
\centering
\begin{tikzpicture}
\begin{axis}[
title={Running times (ms)},
xlabel={$n$},
grid=major,
cycle list={
  {\dypcol,mark=\dyp},
  {\hsbcol,mark=\hsb},
  {\hsccol,mark=\hsc}
},
legend entries={dyn.\ prog.\ two spikes,HS $\lambda=2$ two spikes,HS $\lambda=3$ two spikes},
legend columns=3,
legend to name=two,
small,
ymin=0,
]
\addplot table {twob.dat};
\addplot table {twoc.dat};
\addplot table {twod.dat};
\end{axis}
\end{tikzpicture}
\begin{tikzpicture}
\begin{axis}[
title={Running times (ms)},
xlabel={$n$},
grid=major,
cycle list={
  {\dypcol,mark=\dyp},
  {\hsbcol,mark=\hsb},
  {\hsccol,mark=\hsc}
},
small,
ymin=0,
]
\addplot table {logtwob.dat};
\addplot table {logtwoc.dat};
\addplot table {logtwod.dat};
\end{axis}
\end{tikzpicture}
\\
\ref{two}
\caption{Comparison of running times of dynamic programming and Algorithm~\ref{headalgo} (HS) for two spikes,
means of 100 repeats per data point. 
Left: $\Delta=\lfloor\tfrac{1}{2}\sqrt n\rfloor=k$. Right: $\Delta=40$, $k=\lfloor\log_2(n)\rfloor$.}\label{twofig}
\end{figure}
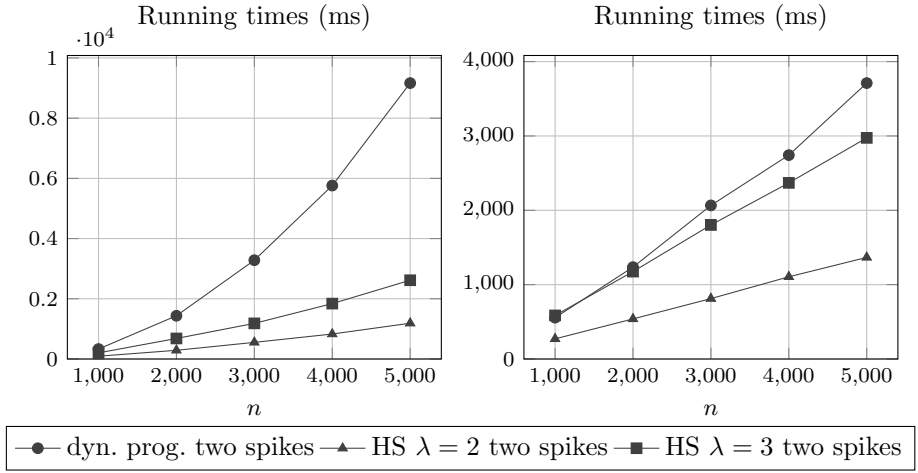

The final two figures contain the test results for \prob{$2$-Spikes Separated Sparsity}. In Figure~\ref{twofig},
the running time of the dynamic programming algorithm is compared to the running time of 
the algorithm of Theorem~\ref{thm:2-head}
for two different precisions. On the left, as in Figure~\ref{runningfig}, we have 
set $k=\Delta=\lfloor\tfrac{1}{2}\sqrt n\rfloor$ and on the right we have set $k=\lfloor\log_2 n\rfloor$ and $\Delta=40$.
The running times behave in a similar way as in Figure~\ref{runningfig}: the approximation algorithm is 
substantially faster, especially when $k=\lfloor\tfrac{1}{2}\sqrt n\rfloor$. Since we only prove a constant
head approximation guarantee in Theorem~\ref{thm:2-head}, we show only the head approximation in Figure~\ref{uniqualifig},
again for uniform as well as Poisson instances. (There is a second reason: in the Poisson instance, even 
for expected arrivial times of $\tfrac{1}{2}\Delta$, occasionally the instance is already $2$-spikes $\Delta$-separated, 
which means that the optimal tail value is~$0$; then, however, we would need to divide by~$0$ in order to compute 
the tail approximation factor.) Again, we see very good head approximation in the case of uniform instances,
and reasonable head approximation in the case of Poisson instances, at least for $\lambda=3$.

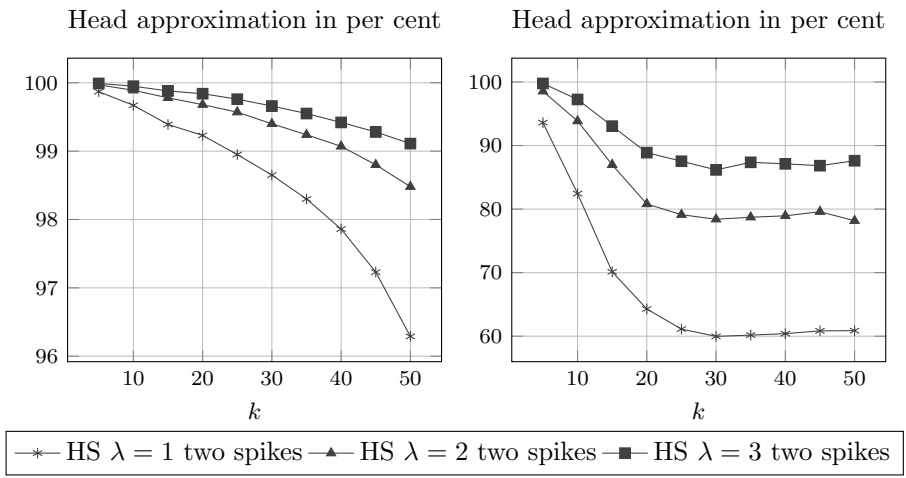
\begin{figure}[!htb]
\centering
\begin{tikzpicture}
\begin{axis}[
title={Head approximation in per cent},
xlabel={$k$},
grid=major,
legend entries={HS $\lambda=1$ two spikes, HS $\lambda=2$ two spikes,HS $\lambda=3$ two spikes},
legend columns=3,
cycle list={
  {\hsacol,mark=\hsa},
  {\hsbcol,mark=\hsb},
  {\hsccol,mark=\hsc}
},
legend to name=twounihead,
small,
]
\addplot table {2unia.dat};
\addplot table {2unib.dat};
\addplot table {2unic.dat};
\end{axis}
\end{tikzpicture}
\begin{tikzpicture}
\begin{axis}[
title={Head approximation in per cent},
xlabel={$k$},
cycle list={
  {\hsacol,mark=\hsa},
  {\hsbcol,mark=\hsb},
  {\hsccol,mark=\hsc}
},
grid=major,
small,
]
\addplot table {2poiheada.dat};
\addplot table {2poiheadb.dat};
\addplot table {2poiheadc.dat};
\end{axis}
\end{tikzpicture}
\\
\ref{twounihead}
\caption{Head approximation for two spikes, in per cent, means of 100 repeats per data point.
Left: Uniform instance with $\Delta=20$ and $n=1000$. 
Right: Poisson process instance with expected arrival time $\Delta=20$ and $n=1000$.}\label{2qualifig}
\end{figure}

\section{Future work}

%\mymargin{comment}
%Ich wei\ss, dass CS-Typen gerne so einen letzten Abschnitt (Conclusion/Future work) m\"ogen, aber brauchen 
%wir das wirklich? Ist doch nur blabla.
% Finde die Fragen schon legitim. Können wir ja in der Journalversion rausschmeissen.

We presented linear time head and tail approximations of arbitrary precision for the $\Delta$-separated model projection problem.
In principle, one might hope for a linear time algorithm to solve the problem to optimality.
We suspect, however, that this is not possible. 
Perhaps a conditional lower bound on the running time could be obtained?
%one can prove a super-linear conditional lower bound on the running time.\mymargin{speculation!}
%In case the dynamic programming algorithm is optimal, one might be able to show that there is no exact algorithm running in time than $O(nk^{1-\delta})$, $\delta >0$, unless some reasonable lower bound for a more common problem fails.

Another question  for future work is whether our results can be extended to 
yield fast algorithms for more general model projection problems.
Here one might start by investigating problems that have a totally unimodular restriction matrix 
 and models that correspond to very limited graph classes like interval graphs.

\bibliographystyle{amsplain}
\bibliography{sparse}

\end{document}